%% file: as140429.tex
\documentclass{article}
\usepackage{lmodern}
\usepackage[T1]{fontenc}

\usepackage{color}
\usepackage{sgame}

\usepackage{thanks}
\usepackage{latexsym}
\usepackage{amssymb}
\usepackage{amsmath}

\usepackage{amstext}
\usepackage{amsopn}

\usepackage{enumerate}
\usepackage{graphics}
\usepackage[all]{xy}
\usepackage{url}
\usepackage{enumerate}

\input{newmacros-a.tex}
\input{macros-apt.tex}

\title{A Classification of Weakly Acyclic Games \thanks{A preliminary version of this paper appeared as \cite{AS12}.}}

\author{
Krzysztof R. Apt
    \thanks{%
      Centre for Mathematics and Computer Science (CWI)
      and University of Amsterdam, The Netherlands
    }
\and
Sunil Simon
    \thanks{%
Department of Computer Science and Engineering,
IIT Kanpur, Kanpur, India
}
}

\date{}
\begin{document}

\maketitle

\begin{abstract}
  Weakly acyclic games form a natural generalization of the class of
  games that have the finite improvement property (FIP). In such games
  one stipulates that from any initial joint strategy some finite
  improvement path exists. We classify weakly acyclic games using the
  concept of a scheduler introduced in \cite{SA12}. We also
  show that finite games that can be solved by the iterated
  elimination of never best response strategies are weakly acyclic.
  Finally, we explain how the schedulers allow us to improve the
  bounds on finding a Nash equilibrium in a weakly acyclic game.
\end{abstract}

\section{Introduction}

\subsection{Background}

Given a strategic game, when we allow the players to improve their
choices on a unilateral basis, we are naturally brought to the concept
of an \emph{improvement path}, in which at each stage a single player
who did not select a best response is allowed to select a better
strategy.  By definition every finite improvement path terminates in a
Nash equilibrium.  This suggests the \emph{finite improvement
  property} (FIP), introduced in \cite{MS96}, according to which every
improvement path is finite. It is obviously a desired property of a
game that in particular is satisfied by the congestion games, 
as explained in \cite{MS96}.

However, the FIP is a very strong property and many natural games do
not satisfy it.  In particular, \cite{Mil96} studied the congestion
games in which the payoff functions are players specific.
These games do not have the FIP.  Milchtaich proved that such games belong to
a larger class of games (essentially introduced in \cite{You93}),
called \emph{weakly acyclic games}, in which one only stipulates that
from any initial joint strategy some finite improvement path exists.

Weakly acyclic games have a natural appeal because the concept of an
improvement path captures the idea of a possible `interaction'
resulting from players' preference for better strategies and hence
provides a link with distributed computing.  In particular,
\cite{ES11} introduced a natural class of weakly acyclic games, which
model the routing aspects on the Internet.  In turn, \cite{MardenAS07}
showed that for weakly acyclic games, a modification of the traditional
no-regret algorithm yields almost sure convergence to a pure Nash
equilibrium. Further, \cite{FabrikantJS10} proved that the existence
of a unique (pure) Nash equilibrium in every subgame implies that the
game is weakly acyclic.

\subsection{Our work}

If we view a strategic game as a distributed system in which the
players attempt to find a Nash equilibrium by means of a `better
response (respectively, `best response) dynamics', then the property
of being weakly acyclic only guarantees that finding a Nash
equilibrium is always possible. However, such an existence guarantee
does not help the players to find it. By adding to the game a
\bfe{scheduler}, a concept introduced in \cite{SA12}, we
ensure that the players always reach a Nash equilibrium, by repeatedly
interacting with it. A scheduler is simply a function that given a
finite sequence of joint strategies selects a player who can improve
his payoff in the last joint strategy. Each player interacts with the
scheduler by submitting to it a strategy he selected. Subsequently the
scheduler again selects a player who did not select a best response.
This interaction process leaves open how each player selects his
better (respectively, best) strategy.

In the presence of a scheduler for a strategic game $G$ we can view
the resulting interaction as a `supergame' between the central
authority represented by the scheduler and the players of $G$.  The
aim of the central authority is to reach a Nash equilibrium in spite
of a limited guarantee on the behaviour of the players: all it can be
sure of is that each selected player will select a better response
(respectively a best response).  The resulting interaction results in
an improvement path (respectively a best response improvement
path). If all so generated improvement paths are finite, we say that
the game \emph{respects} the scheduler.

By providing a classification of the schedulers we obtain a natural
classification of weakly acyclic games.  An advantage of such a
classification is that given a weakly acyclic game we can determine
under what adverse circumstances a Nash equilibrium still can be
reached. Consequently some existing results can be improved.  In
particular, we show in Section \ref{sec:coordination} how we can
strengthen our result from \cite{SA12} concerning a class of
social network games. In turn, \cite{BV12} strengthened the
above mentioned theorem of \cite{Mil96} by showing that congestion
games with player specific payoff functions respect every local best
response scheduler, defined below in Section \ref{sec:schedulers}.
Further, as we explain in Section \ref{sec:bounds}, the existence of a
scheduler allows us to improve bounds on finding a Nash equilibrium.

In turn, a recent paper \cite{KL13} demonstrates the applicability of
schedulers in order to improve bounds on finding a Nash equilibrium in
certain classes of network creation games proposed in \cite{FabNC03}.
The authors show that while an arbitrary improvement path for the
class of so-called MAX swap games on trees might take $O(n^3)$ steps
to converge to an Nash equilibrium, a natural scheduler reduces the
complexity to $O(n \log n)$ and moreover this bound is optimal.

In what follows we introduce eight natural categories of schedulers.
They yield a classification of finite weakly acyclic games that
consists of nine classes games that for two player games collapse into
five classes. Then we study finite games that can be solved by the
iterated elimination of never best response strategies (IENBR) and
show that they are weakly acyclic. We also locate where such games fit
in our classification.  Further, we define a modified notion of a
potential using which we characterize the games that respect a
scheduler.  We also use this notion to reason about a natural class of
coordination games.  Finally, we discuss the consequences of
a fact that a game respects a scheduler.

\section{Preliminaries}
\label{sec:prelim}

Assume a set $N := \{1, \ldots, n\}$ of players, where $n > 1$.  A
\bfe{strategic game} for $n$ players, written as 
$(S_1, \ldots, S_n, p_1, \ldots, p_n)$, consists of a non-empty 
set $S_i$ of
\bfe{strategies} and a \bfe{payoff function} $p_i : S_1 \times \cdots
\times S_n \myra \mathbb{R}$, for each player $i$.  

Fix a strategic game $ G := (S_1, \ldots, S_n, p_1, \ldots, p_n).  $
We denote $S_1 \times \cdots \times S_n$ by $S$, call each element $s
\in S$ a \bfe{joint strategy}, denote the $i$th element of $s$ by
$s_i$, and abbreviate the sequence $(s_{j})_{j \neq i}$ to $s_{-i}$.
Occasionally we write $(s_i, s_{-i})$ instead of $s$.  Finally, we
abbreviate $\times_{j \neq i} S_j$ to $S_{-i}$.

We call a strategy $s_i$ of player $i$ a \bfe{best response} to a
joint strategy $s_{-i}$ of his opponents if $ \fa s'_i \in S_i
\ p_i(s_i, s_{-i}) \geq p_i(s'_i, s_{-i}).  $ If $s_i$ is (not) a best
response to $s_{-i}$, we say that player $i$ \bfe{selected} (\bfe{did
  not select}) \bfe{a best response in $s$}.  Next, we call a joint
strategy $s$ a (pure) \bfe{Nash equilibrium} if each $s_i$ is a best response
to $s_{-i}$, that is, if
$\fa i \in N \ \fa s'_i \in S_i \ p_i(s_i, s_{-i}) \geq p_i(s'_i, s_{-i})$.
We also define
\[
\brset(s) := \{i \mid \textrm{player $i$ selected a best response in } \strprofile\}, 
\]
\[
\mathit{NBR}(s) := \{i \mid \textrm{player $i$ did not select a best response in } \strprofile\}.
\]

Further, we call a strategy $s'_i$ of player $i$ a \bfe{better
  response} given a joint strategy $s$ if $p_i(s'_i, s_{-i}) >
p_i(s_i, s_{-i})$. Following \cite{MS96} a \bfe{path} in $S$ is a
sequence $(s^1, s^2, \LL)$ of joint strategies such that for every $k
> 1$ there is a player $i$ such that $s^k = (s'_i, s^{k-1}_{-i})$ for
some $s'_i \neq s^{k-1}_{i}$.  A path is called an \bfe{improvement
  path} (respectively, a \bfe{best response improvement path}, in
short a \bfe{BR-improvement path}) if it is maximal and for all $k >
1$, $p_i(s^k) > p_i(s^{k-1})$ (respectively, $s^k_i$ is a best
response to $s^{k-1}_{-i}$), where $i$ is the player who deviated from
$s^{k-1}$.  So in an improvement path each deviating player selects a
better response, while in a BR-improvement path each deviating player
selects a best response. Given a finite prefix $\rho$ of a path in
$S$, we denote by $\rho^*$ the infinite path generated by the repeated
concatenation of $\rho$. 

The \bfe{better response graph} (respectively, the \bfe{best response
  graph}) associated with the game $G$ is defined as $(S, \to)$, where
$s \to s'$ if $(s, s')$ is a step in an improvement path
(respectively, in an BR-improvement path).

Given joint strategies $\strprofile, \strprofile' \in S$ 
and a player $i$ we define
\[
\mbox{$\strprofile \betredge{i} \strprofile'$ iff
$\strprofile_{-i} = \strprofile'_{-i}$ and $\payoff_i(\strprofile') >
\payoff_i(\strprofile)$},
\]
\[
\mbox{$\strprofile \bredge{i} \strprofile'$ iff
$\strprofile \betredge{i} \strprofile'$ 
and $\strprofile_{i}'$ is a best response to 
$\strprofile'_{-i}$}.
\]

Recall that $G$ has the \bfe{finite improvement property} (\bfe{FIP}),
(respectively, the \bfe{finite best response property} (\bfe{FBRP}))
if every improvement path (respectively, every BR-improvement path) is
finite. Obviously, if a game has the FIP or the FBRP, then it has a
Nash equilibrium ---it is the last element of each path.  Following
\cite{You93,Mil96} we call a strategic game \bfe{weakly acyclic}
(respectively, \bfe{BR-weakly acyclic}) if for any joint strategy
there exists a finite improvement path (respectively, BR-improvement
path) that starts at it.

In Section \ref{sec:remaining} we shall combine two $n$ players games
$G := (S_1, \ldots, S_n, p_1, \ldots, p_n)$ and
$G':= (S'_1, \ldots, S'_n, p'_1, \ldots, p'_n)$ such 
that $S_1 \cap S'_1 = \ES$ by means of the following 
operation:
\[
G \uplus G' := (S_1 \cup S'_1, \ldots, S_n \cup S'_n, r_1, \ldots, r_n),
\]
where
\[
r_i(s) :=   
\begin{cases}
p_i(s) & \textrm{if } s \in S_1 \times \cdots \times S_n \\
p'_i(s) & \textrm{if } s \in S'_1 \times \cdots \times S'_n \\
0 & \textrm{otherwise}
\end{cases}
\]
The $\uplus$ operation is obviously associative, so it is justified to use $\uplus$ as an $k$-ary operator
on $n$ players games.

\section{Schedulers}
\label{sec:schedulers}

In what follows we introduce some classes of weakly acyclic
games. They are defined in terms of schedulers.  By a \bfe{scheduler}
we mean a function $f$ that given finite sequence $\strprofile^1,\LL,
\strprofile^k$ of joint strategies that does not end in a Nash
equilibrium selects a player who did not select in $\strprofile^k$ a
best response.  In practice schedulers will be applied only to initial
prefixes of improvement paths.

Consider an improvement path
$\rho=(\strprofile^1,\strprofile^2,\ldots)$.  We say that $\rho$
\bfe{respects} a scheduler $f$ if for all $k$ smaller than the
length of $\rho$ we have
$\strprofile^{k+1}=(\strprofile_i',\strprofile^{k}_{-i})$, where
$f(\strprofile^1,\LL, \strprofile^k)=i$.  We say that a strategic
game \bfe{respects a scheduler $f$} if all improvement paths $\rho$
which respect $f$ are finite.
Further, we say that a strategic 
game \bfe{respects a BR-scheduler $f$} if all BR-improvement paths $\rho$
which respect $f$ are finite.

In what follows we study various types of schedulers.
We say that a scheduler $f$ is \bfe{state-based}
if for some function $g: S \myra N$
we have 
\[
f(\strprofile^1,\LL, \strprofile^k)= g(\strprofile^k).
\]

We say that a function
$g: {\cal P}(N) \myra N$ is a \bfe{choice function} if
for all $A \neq \ES$ we have $g(A) \in A$.
Next, we say that a scheduler $f$ is \bfe{set-based}
if for some choice function $g: {\cal P}(N) \myra N$ 
\[
f(\strprofile^1,\LL, \strprofile^k)= g(\mathit{NBR}(\strprofile^k)).
\]
Finally, we say that a set-based scheduler $f$ is \bfe{local} if
the above choice function $g$ satisfies the following
property:
\begin{equation}
  \label{equ:local}
\mbox{if $g(A) \in B \sse A$ then $g(A) = g(B)$.}  
\end{equation}
  
A simple way of producing choice functions $g: {\cal P}(N) \myra N$
that satisfy (\ref{equ:local}) is the following.
Take a permutation $\pi$ of $1, \LL, n$ and define
for $A \neq \ES$
\[
[\pi](A) := \pi(k),
\]
where $k$ is the smallest element of $N$ such that $\pi(k) \in A$.
That is, $[\pi](A)$ is the first element on the list 
$\pi(1), \LL, \pi(n)$ that belongs to $A$.

To simplify notation, we often view a set based scheduler as $f:{\cal
  P}(N) \myra N$ by equating it to the corresponding choice
function. In Section \ref{sec:coordination} we shall need the
following characterization result.

\begin{proposition}
\label{prop:local-permutation}
A choice function $g: {\cal P}(N) \myra N$ satisfies (\ref{equ:local}) iff
it is of the form $[\pi]$ for some
permutation $\pi$ of $1, \LL, n$.
\end{proposition}

\begin{proof}
Suppose a choice function $g: {\cal P}(N) \myra N$ satisfies (\ref{equ:local}).
Define a permutation $\pi$ of $1, \LL, n$ inductively as follows:
\[
\pi(1) := g(N), \: \pi(2) := g(N \setminus \{\pi(1)\}), \LL, \: \pi(n) := g(N \setminus \{\pi(1), \LL, \pi(n-1)\}).
\]

Take now a nonempty subset $A$ of $N$. Let $\pi(k) = [\pi](A)$.  By
definition $\{\pi(1), \LL, \pi(k-1)\} \cap A = \ES$ and $\pi(k) \in
A$.  Let $B := N \setminus \{\pi(1), \LL, \pi(k-1)\}$.
By definition $g(B) = \pi(k)$. Further, $A \sse B$ and $\pi(k) \in A$, 
so by property 
(\ref{equ:local}) we have $g(A) = g(B) = [\pi](A)$.

Next, it is straightforward to check that each function $[\pi]$ 
satisfies (\ref{equ:local}).
\HB
\end{proof}

The games that respect schedulers satisfy obvious
implications that we summarize in Figure~\ref{fig:impl}.  Here
\emph{FIP} (respectively, \emph{FBRP}) stands for the class of games
that have the FIP (respectively, FBRP), \emph{WA} (respectively,
\emph{BRWA}) for the class of weakly acyclic games (respectively,
BR-weakly acyclic games), \emph{Sched} (respectively,
$\mathit{Sched}_{\mathit{BR}}$) stands for the class of games that respect
a scheduler (respectively, a BR-scheduler), etc. In what follows we
clarify which implications can be reversed.

\begin{figure}[ht]
\centering
$
\def\objectstyle{\small}
\def\labelstyle{\small}
\xymatrix@C=12pt@R=15pt{
\mathit{FIP} \ar@{=>}[r] & \mathit{Local} \ar@{=>}[r] \ar@{=>}[d]&\mathit{Set} \ar@{=>}[r]  \ar@{=>}[d]& \mathit{State} \ar@{=>}[r] \ar@{=>}[d]& \mathit{Sched} \ar@{=>}[r] \ar@{=>}[d] & \mathit{WA}  \ar@{<=}[d]\\
\mathit{FBRP} \ar@{=>}[r] \ar@{<=}[u]& \mathit{Local}_{\bestr}  \ar@{=>}[r] &\mathit{Set}_{\bestr} \ar@{=>}[r] & \mathit{State}_{\bestr} \ar@{=>}[r] & \mathit{Sched}_{\bestr} \ar@{=>}[r]& \mathit{BRWA}\\
}$
\caption{\label{fig:impl}Dependencies between various classes of
  weakly acyclic games}
\end{figure}
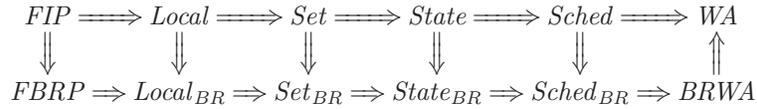

\section{Schedulers versus state-based schedulers}

We prove here three implications which show that the classes of games
$\mathit{Sched}_{\mathit{BR}}$, \emph{Sched}, \emph{State}
and $\mathit{State}_{\mathit{BR}}$ coincide.

\begin{theorem}[\textit{Sched} $\Rightarrow$ \textit{State}]
\label{thm:sch-positional}
If a game respects a scheduler, then it respects a state-based scheduler.
\end{theorem}

\begin{proof}
Fix a finite strategic game $\strgame = (S_1, \ldots, S_n, p_1, \ldots, p_n)$.
Let $Y :=\cup_{k \in \nat} Y_k$, where

\begin{itemize}

\item $Y_0 :=\{\strprofile \in S \mid \mbox{$s$ is a Nash equilibrium}\}$,

\item $Y_{k+1} :=Y_k \cup \{\strprofile \mid \exists i \: \forall \strprofile' (\strprofile \betredge{i} \strprofile' \Rightarrow \strprofile' \in Y_k)\}$.
\end{itemize}

For each $\strprofile \in Y_{k+1} \setminus Y_k$, let
$
\mbox{$\sch_{\st}(\strprofile) :=i$, where $i$ is such that $\forall
  \strprofile' (\strprofile \betredge{i} \strprofile' \to \strprofile'
  \in Y_k)$.}
$

\smallskip

\NI If $Y=S$, then we can view $\sch_{\st}$ as a state-based scheduler.
We now prove two claims concerning the set $Y$ and the scheduler 
$\sch_{\st}$.

\begin{myclaim}
\label{lm:sch-positional2}
If a strategic game $\strgame$ respects a scheduler, then $Y=S$.
\end{myclaim}

\begin{proof}
Suppose that $\strgame$ respects a scheduler $\sch$. Assume by contradiction that
$Y \neq S$. Take $\strprofile^0 \in S \setminus Y$. 
Suppose $\sch(\strprofile^0)=i_1$. 
By the definition of $Y$
there exists a joint strategy $\strprofile^{1}$ such that $\strprofile^0
\betredge{i_1} \strprofile^{1}$ and $\strprofile^{1} \in S \setminus
Y$. Suppose $\sch((\strprofile^0, \strprofile^{1}))=i_2$. Again, by the
definition of $Y$ there exists a joint strategy $\strprofile^{2}$ such that
$\strprofile^{1} \betredge{i_2} \strprofile^{2}$
and $\strprofile^{2} \in S \setminus Y$. Iterating this argument
we construct an infinite improvement path which respects $\sch$,
which yields a contradiction.
\HB
\end{proof}

\begin{myclaim}
\label{lm:sch-positional1}
If for a strategic game $\strgame$ we have $Y=S$, then $\strgame$ respects $\sch_{\st}$.
\end{myclaim}
\begin{proof}
We prove by induction on $k$ that all improvements paths that start in a joint
strategy from $Y_k$ and respect $\sch_{\st}$ are finite.

The claim holds vacuously for $k = 0$. Suppose it holds for some $k \geq 0$. 
Take some $\strprofile \in Y_{k+1} \setminus Y_k$ and an 
improvement path $\xi$ that respects $\sch_{\st}$ and starts in 
$\strprofile$.
Suppose that $\sch_{\st}(\strprofile) :=i$. Then for some $\strprofile'$, 
$\strprofile \betredge{i} \strprofile'$ is the first step in $\xi$.
By the definition of $\sch_{\st}$, 
$\strprofile' \in Y_k$, so by the induction hypothesis $\xi$ is finite.
\HB
\end{proof}

Suppose now that a game $\strgame$ respects a scheduler. By Claim
\ref{lm:sch-positional2} $Y=S$, so $\sch_{\st}$ is a state-based
scheduler.  By Claim \ref{lm:sch-positional1}, $\strgame$ respects
$\sch_{\st}$.

The proof for arbitrary games is a minor modification of the above
argument and uses a transfinite induction to construct the set of
strategies $Y$ and to reason about it. We omit the details.
\HB
\end{proof}

The above proof uses a construction similar to the one used to compute the
winning regions of reachability games, see, e.g., \cite[page 104]{Gra11}.

\begin{theorem}[$\textit{Sched}_{\textit{BR}}$ $\Rightarrow$ $\textit{State}_{\textit{BR}}$]
\label{thm:sch-positionalBR}
If a game respects a BR-scheduler, then it respects a state-based BR-scheduler.
\end{theorem}
\begin{proof}
  The proof is the same as that of Theorem~\ref{thm:sch-positional} with
the relation $\bredge{i}$ used instead of 
$\betredge{i}$.
\HB
\end{proof}

\begin{theorem}[$\textit{Sched}_{\textit{BR}}$ $\Rightarrow$ \textit{Sched}]
\label{thm:sch-BR}
If a finite game respects a BR-scheduler, then it respects a scheduler.
\end{theorem}
\begin{proof}
The idea of the proof is as follows. Suppose that a game respects a BR-scheduler $f_{\textit{BR}}$.
We construct then a scheduler $f$ inductively by repeatedly scheduling the same
player until he plays a best response, and subsequently scheduling the same player
as $f_{\textit{BR}}$ does.

To make it precise we need some notation.  
We call an initial prefix of an improvement path an \bfe{improvement sequence}.
To indicate the deviating players at each step of an improvement sequence
$(s^0, \LL, s^k)$ 
we shall write it alternatively as 
\[
s^0 \betredge{i_1} s^1 \betredge{i_2} \LL \betredge{i_k} s^k.
\]

Given an improvement sequence $\xi$ we denote by $[\xi]_{\textit{BR}}$ the
subsequence of it obtained by deleting the joint strategies that do
not result from a selection of a best response.  In general
$[\xi]_{\textit{BR}}$ is not a improvement sequence (for example, it does not
need to be a maximal sequence), but it is if every maximal
subsequence of it of the form $s^0 \betredge{i} s^1 \betredge{i} \LL
\betredge{i} s^m$ ends with a selection of a best response.

Given a finite sequence of joint strategies $\xi$ we denote its last
element by $last(\xi)$ and denote the extension of $\xi$ by a joint
strategy $s$ by $\xi, s$.  We define the desired scheduler $f$
inductively by the length of the sequences.  For a sequence of length
1, so a joint strategy that is not a Nash equilibrium, we put
\[
f(s) := f_{\textit{BR}}(s).
\]

Suppose now that we defined $f$ on all sequences of length $k$.
Consider a sequence $\xi, s$ of length $k+1$. If $\xi, s$ is not an
improvement path or $last(\xi) \betredge{f(\xi)} s$ does not hold, then we
define $f(\xi, s)$ arbitrarily. Otherwise we put
\[
f(\xi, s) :=
\begin{cases}
  f_{\textit{BR}}([\xi, s]_{\textit{BR}}) & \textrm{if $s_i$ is a best response to $s_{-i}$} \\
  f(\xi)          & \textrm{otherwise}
\end{cases}
\]

We claim that $G$ respects the scheduler $f$. To see it take an
improvement path $\xi$ that respects $f$.  By the definition of $f$,
$[\xi]_{\textit{BR}}$ is an improvement sequence that respects $f_{\textit{BR}} $.  By
assumption $[\xi]_{\textit{BR}}$ is finite, so $\xi$ is finite, as well.
\HB
\end{proof}

Note that the above theorem fails to hold for infinite games.  Indeed,
consider a two player game $(\{0\}, [0,1], p_1, p_2)$, where $[0,1]$
stands for the real interval $\{r \mid 0 \leq r \leq 1\}$ and
$p_1(0,s_2) = p_2(0,s_2) := s_2$.  Then 1 is a unique best response of
player 2 to the strategy 0, so this game respects the unique
BR-scheduler. However, it does not respect the unique scheduler.

\section{Two player games}
\label{sec:two}

For two player games more implications hold.

\begin{proposition}[$\textit{Sched}$ $\Rightarrow$ \textit{FBRP}]
  If a two player game respects a scheduler, then
  every BR-improvement path is finite.
\end{proposition}

\begin{proof}
Suppose that a two player game $G$ respects a scheduler. Note that
the best response graph of $G$ has the property that for every node
$s$ that is not a source node, the set $\mathit{NBR}(s)$ has at most
one element. Take a BR-improvement path $\xi$. Suppose that $(s, s')$
is the first step in $\xi$ and that $\eta$ is the suffix of $\xi$ that
starts at $s'$.  Then every element $s''$ of $\eta$ is such that the
set $\mathit{NBR}(s'')$ has at most one element.  Hence $\eta$
respects any scheduler and consequently is finite. So $\xi$ is finite,
as well.  \HB
\end{proof}

We also have the following examples showing that other implications do
not hold.

\begin{myexample}[\textit{Local} $\not \Rightarrow$ \textit{FIP}]
\rm
Consider the following game

\begin{center}
 \begin{game}{2}{3}
     & $L$    & $C$  & $R$\\
$T$   &$1,0$   &$0,1$ &$0,2$\\
$B$   &$0,1$   &$1,0$ &$0,0$\\
\end{game}
\end{center}

It respects the local scheduler $f$ for which $f(\{1,2\}) = 2$.
However, this game does not have the FIP. 
\HB
\end{myexample}

\begin{myexample}[\textit{State} $\not \Rightarrow$ \textit{Set}, 
\textit{FBRP} $\not \Rightarrow$ \textit{Set}]
\label{exa:1}
\rm 
Consider the game 

\begin{center}
  \begin{game}{3}{3}
      & $A$    & $B$  & $C$\\
$A$   &$2,2$   &$2,0$ &$0,1$\\
$B$   &$0,2$   &$1,1$ &$1,0$\\
$C$   &$1,0$   &$0,1$ &$0,0$
\end{game} 
\end{center}

In Figure~\ref{fig:1} we display
the better response graph and the best response graphs
associated with it.  
 
\VV

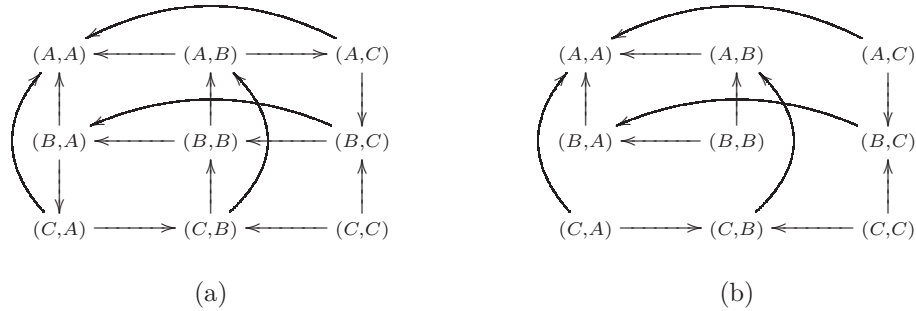
\begin{figure}[ht]
\centering
\begin{tabular}{ccc}
$
\def\objectstyle{\scriptstyle}
\def\labelstyle{\scriptstyle}
\xymatrix@R=20pt @C=30pt{
(A,A) &(A,B) \ar[l] \ar[r] &(A,C) \ar[d] \ar@/_1.5pc/[ll]\\
(B,A) \ar[u] \ar[d]&(B,B) \ar[l] \ar[u] &(B,C) \ar[l] \ar@/_1.3pc/[ll] \\
(C,A)  \ar@/^1.5pc/[uu] \ar[r]&(C,B) \ar@/_1.8pc/[uu] \ar[u] &(C,C) \ar[u] \ar[l]\\
}$
& ~~~~~~~~~~~&
$
\def\objectstyle{\scriptstyle}
\def\labelstyle{\scriptstyle}
\xymatrix@R=20pt @C=30pt{
(A,A) &(A,B)  \ar[l] &(A,C) \ar[d] \ar@/_1.5pc/[ll]\\
(B,A) \ar[u] &(B,B) \ar[l] \ar[u] &(B,C)  \ar@/_1.3pc/[ll] \\
(C,A)  \ar@/^1.5pc/[uu] \ar[r]&(C,B) \ar@/_1.8pc/[uu] &(C,C) \ar[u] \ar[l]\\
}$
\\
\\
(a)  & &(b)\\
 \end{tabular}
\caption{\label{fig:1}Better response and best response graphs}
\end{figure}

First we show that this game 
respects a state-based scheduler. To define it suffices to consider
the joint strategies in 
which none of the players selected a best response. These are
$(A,C), \ (C,A)$, \ $(B,B)$ and $(C,C)$.
We put
\[
f(A,C) := 2, \ f(C,A) := 1, \ f(B,B) := 1, \ f(C,C) := 1.
\]
($f(B,B) := 2$ and $f(C,C) := 2$ would do, as well.)

Using Figure~\ref{fig:1}(a) it is
easy to check that any improvement path that respects this scheduler
ends in $(A,A)$.  Further, the graph given in Figure~\ref{fig:1}(b) is
acyclic, that is, this game has the FBRP.

However, this game does not respect any set-based scheduler.
Indeed, suppose otherwise. Then such a scheduler is defined using a choice function
$g$. If $g(\{1, 2\}) = 1$, then the infinite improvement path 
$((B,B), \ (A,B), \ (A,C), \ (B,C))^{*}$ respects this scheduler.
In turn, if $g(\{1, 2\}) = 2$, then the 
infinite improvement path $((B,B), \ (B,A), \ (C,A), \ (C,B))^{*}$ 
respects this scheduler.
\HB
\end{myexample}

\begin{myexample}[\textit{BRWA} $\not\Rightarrow$ $\textit{Sched}_{\textit{BR}}$]
\rm
Consider the following game

\begin{center}
 \begin{game}{2}{3}
     & $L$    & $C$  & $R$\\
$T$   &$0,1$   &$1,0$ &$0,1$\\
$B$   &$1,0$   &$0,1$ &$0,0$\\
\end{game}
\end{center}

It is BR-weakly acyclic. However, this game respects no BR-scheduler.
Indeed, there is only one scheduler for this game and the following
infinite improvement path respects it:
\[((T,L), \ (B,L), \ (B,C), \ (T, C))^*.
\] 

\HB
\end{myexample}

\begin{myexample}[\textit{WA} $\not\Rightarrow$ \textit{BRWA}]
\rm
Consider now the following game

\begin{center}
 \begin{game}{3}{3}
     & $L$    & $C$  & $R$\\
$T$   &$0,2$   &$1,0$ &$0,1$\\
$M$   &$1,0$   &$0,1$ &$0,0$\\
$B$   &$0,0$   &$0,0$ &$1,0$\\
\end{game}
\end{center}

It is weakly acyclic. However, it is not BR-weakly acyclic, 
because the infinite BR-improvement path 
$((T,L), \ (M,L), \ (M,C), \ (T, C))^*$ is a unique BR-improvement path starting at $(T,L)$.
\HB
\end{myexample}

By definition if a two player game respects a set-based scheduler,
then it respects a local scheduler. So putting together the results
obtained so far we get the implications and equivalences
depicted in Figure~\ref{fig:impl2}. All implications are proper.


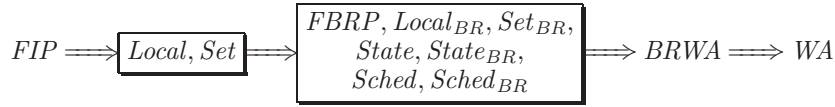
\begin{figure}[ht]
\centering
$
\def\objectstyle{\small}
\def\labelstyle{\small}
\xymatrix@C=20pt{
\mathit{FIP} \ar@{=>}[r] &*+[F-,] {\mathit{Local}, \mathit{Set}} \ar@{=>}[r] 
& *+[F-,] \txt{$\mathit{FBRP}, \mathit{Local}_{\bestr}, \mathit{Set}_{\bestr}$, \\ 
$\mathit{State}, \mathit{State}_{\bestr}$,\\
$\mathit{Sched}, \mathit{Sched}_{\bestr}$} \ar@{=>}[r] 
& \mathit{BRWA} \ar@{=>}[r] & \mathit{WA}
}$
\caption{\label{fig:impl2}Classification of two player finite weakly acyclic games}
\end{figure}

As an illustration consider a two player congestion game with
player-specific payoff functions (actually a class of games) analyzed
in \cite[page 115]{Mil96}. Each player has three strategies numbered
1, 2, and 3. We omit the description of the game and mention only its
relevant characteristics. The game has two Nash equilibria, $(1,2)$
and $(2,1)$, and an infinite improvement path $((3,2), \ (2,2), \
(2,3),$ $(1,3), \ (1,1), \ (3,1))^*$. Additionally $p_1(1,2) > p_1(1,1)$,
$p_1(1,2) > p_1(2,2)$, and $p_2(2,1) > p_2(2,3)$.

This implies that the graph depicted in Figure~\ref{fig:Mil} is a
subgraph of the better response graph of the game.  The dotted edges
are implied by the fact that $(1,2)$ and $(2,1)$ are Nash equilibria,
the continuous edges correspond to the infinite improvement path,
while the dashed edges are implied by the mentioned properties of the
payoff functions. The relation between the nodes $(2,1)$ and $(3,1)$
and the status of edges involving $(3,3)$ are unspecified, though the
edge $(2,1) \to (3,1)$ is excluded since $(2,1)$ is a Nash
equilibrium.

Note that this game respects a set-based scheduler $g$ such that
$g(\{1,2\}) = 2$. Indeed, this choice allows us to `exit' the infinite
improvement path both at $(1,1)$, $(2, 2)$ and $(2,3)$.  One can also
check that this choice properly takes care of any legal completion of
the graph depicted in Figure~\ref{fig:Mil} to a possible better
response graph.  For instance, addition of the edges $(3,2) \to (3,3)$
and $(3,3) \to (2,3)$ would create another infinite improvement path
that would be `exited' at $(1,1)$.  We conclude that each game with
the above characteristics belongs to $\textit{Set} \setminus
\textit{FIP}$.  Figure~\ref{fig:impl2} implies that each such game
also belongs to $\textit{FBRP}$.

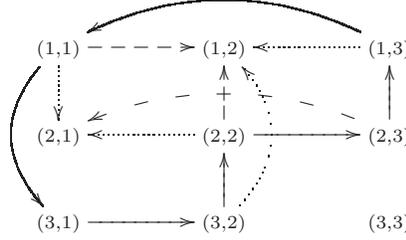
\begin{figure}[ht]
\centering
$
\def\objectstyle{\scriptstyle}
\def\labelstyle{\scriptstyle}
\xymatrix@R=20pt @C=40pt{
(1,1) \ar@{-->}[r] \ar@{.>}[d] \ar@/_1.5pc/[dd]&(1,2) &(1,3) \ar@{.>}[l] \ar@/_1.5pc/[ll]\\
(2,1) &(2,2) \ar@{.>}[l] \ar@{-->}[u] \ar[r] &(2,3) \ar[u] \ar@{-->}@/_1.3pc/[ll] \\
(3,1) \ar[r] &(3,2) \ar@{.>}@/_1.5pc/[uu] \ar[u] &(3,3)\\
}$
\caption{\label{fig:Mil} Subgraph of a better response graph}
\end{figure}

\section{Remaining implications}
\label{sec:remaining}

To deal with the remaining implications we need games with at least
three players.

\begin{myexample}[$\textit{Local}_{\textit{BR}}$ $\not \Rightarrow$ $\textit{FBRP}$]
\label{exa:LnF}
\rm 
We consider a three player game in which every player has two
strategies, 0 and 1, and such that its best response graph is as shown
in Figure~\ref{fig:3a}.

\begin{figure}[ht]
\centering
$
\def\objectstyle{\scriptstyle}
\def\labelstyle{\scriptstyle}
\xymatrix@!0@C=30pt@R=30pt{
& (0,1,1) \ar@{--}'[d][dd] \ar@{->}[dl]
&& (1,1,1) \ar@{->}[ll] \\
(0,0,1) \ar@{->}[rr] & &(1,0,1) \ar@{->}[ur]\ar@{->}[dd]\\
& (0,1,0) \ar@{--}'[r][rr] && (1,1,0) \ar@{--}[uu]\\
(0,0,0) \ar@{--}[rr] \ar@{--}[uu] \ar@{--}[ur] &&(1,0,0) \ar@{--}[ur]
}
$
\caption{\label{fig:3a}A best response graph}
\end{figure}
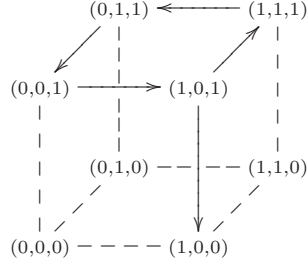

To define the corresponding payoff functions we just need
to interpret each directed edge $(a,b,c) \to (a,b',c)$ as
the statement $p_2(a,b,c) = 0$ and  $p_2(a,b',c) = 1$,
each dotted edge $(a,b,c) - - (a,b',c)$ as
the statement $p_2(a,b,c) = 0$ and  $p_2(a,b',c) = 0$,
and similarly for the other edges.
This yields
\[
p_1(0,1,1) = p_1(1,0,1) = p_2(0,0,1) = p_2(1,1,1) = p_3(1,0,0) := 1,
\]
with the remaining payoffs equal to 0.

This game respects any local BR-scheduler $f$ for which $f(\{2,3\}) = 3$.
However, this game does not have the FBRP. 
\HB
\end{myexample}

\begin{myexample}[$\textit{Set}$ $\not \Rightarrow$ $\textit{Local}$, 
$\textit{Set}_{\textit{BR}}$ $\not \Rightarrow$ $\textit{Local}_{\textit{BR}}$]
\label{exa:2}
\rm We first construct four three player games.  Every player in each
game has two strategies, 0 and 1. The better response graphs of these
games are depicted in Figure~\ref{fig:4}. To define the corresponding
payoff functions we proceed as in Example~\ref{exa:LnF}.

\begin{figure}[ht]
\centering
\begin{tabular}{cc}
$
\def\objectstyle{\scriptstyle}
\def\labelstyle{\scriptstyle}
\xymatrix@!0@C=30pt@R=30pt{
& (0,1,1) \ar@{--}'[d][dd] \ar@{->}[dl]
&& (1,1,1) \ar@{->}[ll] \\
(0,0,1) \ar@{->}[rr] & &(1,0,1) \ar@{->}[ur]\ar@{->}[dd]\\
& (0,1,0) \ar@{--}'[r][rr] && (1,1,0) \ar@{--}[uu]\\
(0,0,0) \ar@{--}[rr] \ar@{--}[uu] \ar@{--}[ur] &&(1,0,0) \ar@{--}[ur]
}
$
&
$
\def\objectstyle{\scriptstyle}
\def\labelstyle{\scriptstyle}
\xymatrix@!0@C=30pt@R=30pt{
& (0,1,1) \ar@{--}'[d][dd] \ar@{--}[dl]
&& (1,1,1) \ar@{--}[ll] \ar@{->}[dl]\\
(0,0,1)  & &(1,0,1) \ar@{->}[ll] \ar@{->}[dd]\\
& (0,1,0) \ar@{--}'[r][rr] && (1,1,0) \ar@{->}[uu]\\
(0,0,0) \ar@{--}[rr] \ar@{--}[uu] \ar@{--}[ur] &&(1,0,0) \ar@{->}[ur]
}
$ \\ \\
$
\def\objectstyle{\scriptstyle}
\def\labelstyle{\scriptstyle}
\xymatrix@!0@C=30pt@R=30pt{
& (0,1,1) \ar@{--}'[d][dd] \ar@{--}[dl]
&& (1,1,1) \ar@{--}[ll]\\
(0,0,1)  \ar@{->}[dd] & &(1,0,1) \ar@{->}[ll] \ar@{->}[ur]\\
& (0,1,0) \ar@{--}'[r][rr] && (1,1,0) \ar@{--}[uu]\\
(0,0,0) \ar@{->}[rr] \ar@{--}[ur] &&(1,0,0) \ar@{->}[uu] \ar@{--}[ur]
}
$ 
&
$
\def\objectstyle{\scriptstyle}
\def\labelstyle{\scriptstyle}
\xymatrix@!0@C=30pt@R=30pt{
& (0,1,1) \ar@{--}'[d][dd] \ar@{--}[dl]
&& (1,1,1) \ar@{--}[ll]\\
(0,0,1)  \ar@{--}[dd] & &(1,0,1) \ar@{->}[ll] \ar@{->}[ur] \ar@{->}[dd]\\
& (0,1,0) \ar@{--}'[r][rr] && (1,1,0) \ar@{--}[uu]\\
(0,0,0) \ar@{--}[rr] \ar@{--}[ur] &&(1,0,0)  \ar@{--}[ur]
}
$
\end{tabular}
\caption{\label{fig:4}Four better response graphs}
\end{figure}
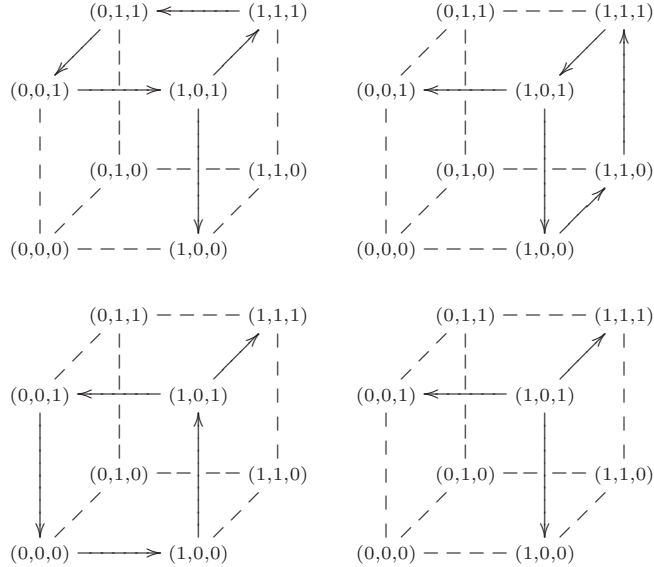

Next, we make for each player his strategy sets in these four games
mutually disjoint by renaming the strategies 0 and 1 in the $i$th game
to $0_i$ and $1_i$.  Denote the resulting games by $G_1, \LL, G_4$.
Let $G := G_1 \uplus G_2 \uplus G_3 \uplus G_4$.  In the better
response graph of $G$ there are precisely three loops that correspond
to the ones depicted in Figure~\ref{fig:4}.  Using this figure one can
check that $G$ respects the following set-based scheduler:
\[
f(\{2, 3\}) := 3, \ f(\{1,3\}) := 1, \ f(\{1,2\}) := 2, 
\]
with an arbitrary value for the input $\{1,2,3\}$.

However, $G$ does not respect any local scheduler, since to `exit' each of these three loops 
each set-based scheduler needs
to make the above selections, and then no selection for the input $\{1,2,3\}$
can make the scheduler local.

The above reasoning also holds for the BR-schedulers, since every
better response in $G$ is also a best response, as all payoffs in $G$
are either 0 or 1.  
\HB
\end{myexample}

\begin{myexample}[$\textit{State}_{\textit{BR}}$ $\not \Rightarrow$ $\textit{Set}_{\textit{BR}}$]
\label{exa:3}
\rm
We first construct two three player games.  Every player in each game
has two strategies, 0 and 1. The best response graphs of these games
are depicted in Figure~\ref{fig:5}.
The first one coincides with the first better response graph of Figure~\ref{fig:4}.

\begin{figure}[ht]
\centering
\begin{tabular}{cc}
$
\def\objectstyle{\scriptstyle}
\def\labelstyle{\scriptstyle}
\xymatrix@!0@C=30pt@R=30pt{
& (0,1,1) \ar@{--}'[d][dd] \ar@{->}[dl]
&& (1,1,1) \ar@{->}[ll] \\
(0,0,1) \ar@{->}[rr] & &(1,0,1) \ar@{->}[ur]\ar@{->}[dd]\\
& (0,1,0) \ar@{--}'[r][rr] && (1,1,0) \ar@{--}[uu]\\
(0,0,0) \ar@{--}[rr] \ar@{--}[uu] \ar@{--}[ur] &&(1,0,0) \ar@{--}[ur]
}
$
&
$
\def\objectstyle{\scriptstyle}
\def\labelstyle{\scriptstyle}
\xymatrix@!0@C=30pt@R=30pt{
& (0,1,1) \ar@{--}'[d][dd] \ar@{--}[dl]
&& (1,1,1) \ar@{--}[ll] \\
(0,0,1) \ar@{->}[rr] & &(1,0,1) \ar@{->}[ur]\ar@{->}[dd]\\
& (0,1,0) \ar@{--}'[r][rr] && (1,1,0) \ar@{--}[uu]\\
(0,0,0) \ar@{->}[uu] \ar@{--}[ur] &&(1,0,0) \ar@{->}[ll] \ar@{--}[ur]
}
$
\end{tabular}
\caption{\label{fig:5} Two best response graphs}
\end{figure}
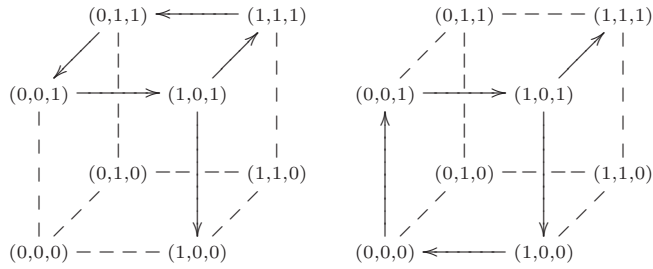

Denote these games by $G_1$ and $G_2$ and let $G := G_1 \uplus G_2$.
Then $G$ respects the following state-based BR-scheduler
(we use here the same renaming of the strategies as in Example~\ref{exa:2}):
\[
f(1_1, 0_1, 1_1) := 3, \ f(1_2, 0_2, 1_2) := 2.
\]
However, $G$ does not respect any set-based BR-scheduler, since in each of the above two joint
strategies the set of players who did not select a best response is the same (namely $\{2,3\}$),
and to `exit' each of the two loops in the corresponding best response graphs
one needs to take the above choices.
\HB
\end{myexample}

This brings us to the following final diagram depicted in Figure~\ref{fig:implf}.
All implications are proper.


\begin{figure}[ht]
\centering
$
\def\objectstyle{\small}
\def\labelstyle{\small}
\xymatrix@C=12pt@R=15pt{
& \mathit{FBRP} \ar@{=>}[dr]\\
\mathit{FIP} \ar@{=>}[ur] \ar@{=>}[dr] & & \mathit{Local}_{\bestr} \ar@{=>}[dr]\\
& \mathit{Local} \ar@{=>}[ur] \ar@{=>}[dr] & &\mathit{Set}_{\bestr} \ar@{=>}[r] 
  &*+[F-,] \txt{$\mathit{State}, \mathit{State}_{\bestr}$,\\ $\mathit{Sched},\mathit{Sched}_{\bestr}$} \ar@{=>}[r]
  & \mathit{BRWA} \ar@{=>}[r] & \mathit{WA}\\
& & \mathit{Set} \ar@{=>}[ur]\\
}$
\caption{\label{fig:implf}Classification of finite weakly acyclic games}
\end{figure}
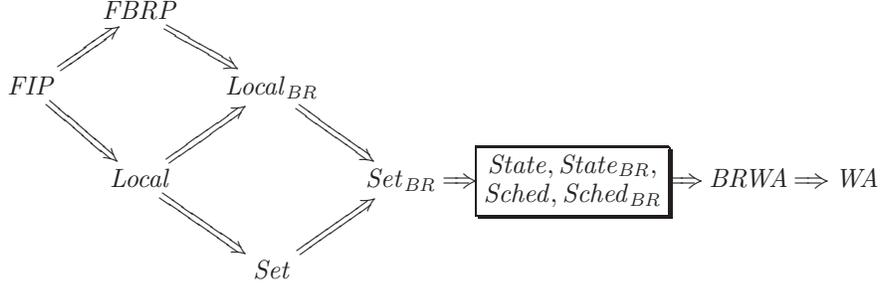

\section{Games that can be solved by the iterated elimination of NBR}

Weakly acyclic games are of natural interest because they have (pure)
Nash equilibria. In the literature another class of finite games has
been studied with the same property, namely those that can be solved
by the iterated elimination of strictly dominated strategies. In what
follows we focus on a more general class of finite games that can be
solved by the iterated elimination of never best responses (IENBR),
where we limit ourselves to never best responses to joint pure
strategies of the opponents.  

In this context the joint pure strategies of the opponents are usually
called \bfe{point beliefs} and the outcome of IENBR is the set of
\bfe{rationalizable joint strategies}, see \cite{Ber84}. In other
words, we focus on finite game that have a unique rationalizable joint
strategy w.r.t.~point beliefs.

Given a strategic game $\strgame=(S_1, \ldots, S_n, p_1, \ldots, p_n)$
we say that a strategy $\strprofile_i \in S_i$ is a \bfe{never best
  response} in $G$ if for every $\strprofile_{-i} \in S_{-i}$ there is a
strategy $\strprofile_i' \in S_i$ such that
$\payoff_i(\strprofile_i',\strprofile_{-i}) >
\payoff_i(\strprofile_i,\strprofile_{-i})$.

To better understand the concept of a never best response recall that
given a finite strategic game $\strgame=(S_1, \ldots, S_n, p_1,
\ldots, p_n)$ a \bfe{mixed strategy} for player $i$ is a probability
distribution over $S_i$. Further, recall that a strategy
$\strprofile_i \in S_i$ is \bfe{strictly dominated} by a mixed
strategy if there exists a mixed strategy $m_i$ such that for all
$\strprofile_{-i} \in S_{-i}$, $\payoff_i(m_i,\strprofile_{-i}) >
\payoff_i(\strprofile_i,\strprofile_{-i})$, where
$\payoff_i(m_i,\strprofile_{-i})$ is defined in the expected way.

Then it is straightforward to check that if $\strprofile_i$ is
strictly dominated by a mixed strategy then $\strprofile_i$ is a never
best response.  Also, the reverse implication does not hold.

Given a binary relation $\to$ we denote by $\tra$ its transitive reflexive closure.
In what follows we consider a specific relation $\myra$ between games defined as follows:
\[
G \myra G'
\]
where $G = (S_1, \ldots, S_n, p_1, \ldots, p_n)$,
$G' = (S'_1, \ldots, S'_n, p'_1, \ldots, p'_n)$ is a proper subgame of $G$, and
\[
\mbox{$\fa i \in \C{1, \ldots, n} \ \fa s_i \in S_i \setminus S'_i \ \neg \te s_{-i} \in S_{-i}$ $s_i$ is a best response to $s_{-i}$ in $S$}.
\]
That is, $G \myra G'$ when $G'$ results from $G$ by removing from it some strategies 
that are never best responses. 

If for some game $G'$ in which each player is left with exactly one
strategy we have $G \tra G'$, then we say that $G$ can be
solved by the iterated elimination of never best responses, in short
\bfe{can be solved by IENBR}.

The games that can be solved by IENBR are of interest because of the following
observation.

\begin{proposition}
  If a finite game can be solved by IENBR then it has a unique Nash
  equilibrium.
\end{proposition}

\begin{proof}
It suffices to note that each step of the elimination process maintains the set
of Nash equilibria.  
\HB
\end{proof}

We now prove that this class of games is actually contained in the class
of weakly acyclic games. More precisely, the following holds.

\begin{theorem}
\label{thm:solv-st}
If a finite game can be solved by IENBR, then it is in
$\mathit{State}_{\bestr}$.
\end{theorem}
\begin{proof}
We use the scheduler $f$ which schedules each time the player who did
not play a best response the longest, breaking ties in favour of the
player with the smallest index. Fix a finite game $\strgame$ that can
be solved by IENBR. By definition of the scheduler $f$, it satisfies
the following property:
\begin{description}
\item [(P1)] for all $i \in N$, for every BR-improvement path
  $\rho=\strprofile^0, \strprofile^1, \LL$ in $\strgame$ that respects
  $f$ and for all $j \geq 0$, there exists $k \geq j$ such that $i \in
  \brset(\strprofile^k)$.
\end{description}

Let $\rho=\strprofile^0,\strprofile^1,\ldots$ a BR-improvement path
which respects $\sch$. We argue that $\rho$ is finite.

We find it useful to introduce the following notion: a finite sequence
of strategy profiles $b=s^0,\ldots,s^k$ is called a \bfe{block} if for
all $i \in \nbrset(s^0)$, there exists $j \in \{0,\ldots,k\}$ such
that either $i \in \brset(s^j)$ or $i \in \sch(s^0,\ldots,s^j)$. The
sequence $b$ is said to be minimal if there is no prefix $b'$ of $b$
such that $b'$ is a block. Since $\sch$ satisfies (P1), we can
represent the best response improvement path $\rho$ as a concatenation
of minimal blocks $b^0,b^1,\ldots$.

Since $\strgame$ can be solved by IENBR, there is a finite reduction
sequence $\xi=\strgame^0 \to \strgame^1 \to \ldots \to \strgame^m$ for
$\strgame$ such that each player is left with a single strategy in
$\strgame^m$. By Theorem 4.2 in \cite{Apt05} we can assume without
loss of generality that each reduction step eliminates a single
strategy. Let $(i_1,\strprofile_{i_1}),
(i_2,\strprofile_{i_2}),\ldots,(i_m,\strprofile_{i_m})$ be the
corresponding sequence of pairs of player identity and of the strategy
which are eliminated in $\xi$. Consider the first pair
$(i_1,\strprofile_{i_1})$ and the block
$b^0=\strprofile^0,\ldots,\strprofile^l$. We have two cases:
\begin{itemize}
\item $\strprofile^0_{i_1}=\strprofile_{i_1}$. 

\noindent Then $i_1 \in \nbrset(\strprofile^0)$. Indeed, since
$\strprofile_{i_1}$ is a never best response, $i_1 \in
\nbrset(\strprofile)$ for any $\strprofile \in S$. Since $\sch$
satisfies (P1) and $\strprofile_{i_1}$ is a never best response there
exists $j \in \{0,\ldots, l\}$ such that $i_1 \in
\sch(s^0,\ldots,s^j)$. Therefore $\strprofile^{j+1}_{i_1} \neq
\strprofile_{i_1}$. Again, since $\strprofile_{i_1}$ is a never best
response player $i$ never switches back to $\strprofile_{i_1}$ in the
suffix of $\rho$ starting in $\strprofile^j$ irrespective of the
strategies chosen by the other players.

\item $\strprofile^0_i \neq \strprofile_{i_1}$.

\noindent Then since $\strprofile_{i_1}$ is a never best response,
player $i_1$ never switches to $\strprofile_{i_1}$ in the improvement
path $\rho$, irrespective of the strategies chosen by the other
players.
\end{itemize}
Therefore in the suffix $\rho^1=b^1,b^2,b^3,\ldots$ of the improvement
path $\rho$ the strategy $\strprofile_{i_1}$ is never chosen. This
means that $\rho^1$ is a best response improvement path in the
restricted game $\strgame^1$. Now consider the pair
$(i_2,\strprofile_{i_2})$ and the block $b^1$. By the same argument as
above, we can show that in the suffix $\rho^2=b^2,b^3,\ldots$ of the
improvement path $\rho$, the strategy $\strprofile_{i_2}$ is never
chosen and therefore $\rho^2$ is a best response improvement path in
the restricted game $\strgame^2$. Since $\strgame$ can be solved by
the elimination sequence $\xi=\strgame^0,\strgame^1,\ldots,
\strgame^m$, by iterating the above argument we get that there is a
suffix of the improvement path $\rho$ which involves only joint
strategies in $\strgame^m$. Since in $\strgame^m$ each player is left
with a single strategy, $\rho$ is finite.

We conclude that the game $\strgame$ respects the BR-scheduler $f$.
Consequently, by Theorem \ref{thm:sch-positionalBR}, $\strgame$ is in
$\mathit{State}_{\bestr}$.  
\HB
\end{proof}

In particular, in view of the remarks at the beginning of this
section, if a game can be solved by iterated elimination of strictly
dominated strategies (by a mixed strategy), in short, IESDS, then it
is weakly acyclic.  An example of such a game is the \bfe{beauty
  contest game} due to \cite{Mou86}. In this game there are $n > 2$
players, each with the set of strategies equal $\{1, \ldots, 100\}$.
Each player submits a number and the payoff to each player is obtained
by splitting 1 equally between the players whose submitted number is
closest to $\frac{2}{3}$ of the average.  For example, if the
submissions are $29,32,29$, then the payoffs are respectively
$\frac{1}{2}, 0, \frac{1}{2}$.  By Theorem~\ref{thm:solv-st}
this game is weakly acyclic.

Note that the implication in the above theorem cannot be reversed. Indeed,
the game given in Example~\ref{exa:3} is in $\textit{State}_{\textit{BR}}$
and it is easy to check that it cannot be solved by IENBR, as every strategy
in this game is a best response. An even simpler example is the well-known
Battle of Sexes game.

We now identify the smallest class used in our classification of
finite weakly acyclic games that includes all finite games that can be
solved by IENBR.  We begin with two player games for which we have the
classification presented in Figure~\ref{fig:impl2} of
Section~\ref{sec:two}.  In view of this, it suffices to note the
following example.

\begin{myexample}[$\textit{IENBR}$ $\not \Rightarrow$ $\textit{Set}$]
\rm 
Consider the game 

\begin{center}
  \begin{game}{3}{3}
      & $A$    & $B$  & $C$\\
$A$   &$3,2$   &$1,1$ &$2,0$\\
$B$   &$1,2$   &$3,1$ &$1,2$\\
$C$   &$2,2$   &$2,3$ &$1,1$
\end{game} 
\end{center}
and its better response graph displayed in Figure~\ref{fig:set}.
\begin{figure}[ht]
\centering
$
\def\objectstyle{\scriptstyle}
\def\labelstyle{\scriptstyle}
\xymatrix@R=20pt @C=30pt{
(3,2) &(1,1) \ar[l] \ar[d] \ar@/_1.5pc/[dd] &(2,0) \ar[l] \ar@/_1.5pc/[ll]\\
(1,2) \ar[u] \ar[d]&(3,1) \ar[l] \ar[r] &(1,2) \ar[u] \\
(2,2)  \ar@/^1.5pc/[uu] \ar[r] &(2,3) \ar[u] &(1,1) \ar@/_1.5pc/[uu] \ar[l] \ar@/^1.5pc/[ll]\\
}$ \\[3mm]
\caption{\label{fig:set} A better response graph}
\end{figure}

It is easy to see that this game can be solved by IENBR and that the outcome is 
$(A,A)$. However, this game does not respect any set-based scheduler.
Indeed, suppose otherwise. Then such a scheduler is defined using a choice function
$g$. If $g(\{1, 2\}) = 1$, then the infinite improvement path 
$((A,B), \ (B,B), \ (B,C), \ (C,C))^{*}$ 
respects this scheduler.
In turn, if $g(\{1, 2\}) = 2$, then the 
infinite improvement path 
$((C,A), \ (C,B), \ (B,B), \ (B,A))^{*}$ 
respects this scheduler.
\HB
\end{myexample}

The case of arbitrary games is more complicated.  First we show how
given two $n$ player games $G_1 := (S_1, \ldots, S_n, p_1, \ldots,
p_n)$ and $G_2 := (S'_1, \ldots, S'_n, r_1, \ldots, r_n)$, where $S_1
\cap S'_1 = \ES$, that have strictly positive payoffs and both can be
solved by IENBR, we can construct using an auxiliary game $G_3$ a game
$G := G_1 \uplus G_2 \uplus G_3$ that also can be solved by IENBR.

Suppose that the outcomes of the IENBR applied to the games $G_1$ and $G_2$ are respectively the joint strategies 
$(s_1, \LL, s_n)$ and $(s'_1, \LL, s'_n)$.
Further, assume that $a \not\in S_1 \cup S'_1$.
Let
$G_3 := (\{a\}, S_2 \cup S'_2, \ldots, S_n \cup S'_n, u_1, \ldots, u_n)$,
where the $u_i$ payoff functions are defined in such a way that
\begin{itemize}
\item the NBR eliminations steps applied to $G_1$ remain valid in the context of $G$,

\item the same holds for the game $G_2$,

\item the resulting game $(\{a, s_1, s'_1\}, \{s_2, s'_2\}, \ldots, \{s_n,
  s'_n\}, u_1, \ldots, u_n)$ can be solved by IENBR.

\end{itemize}

The first two items can be achieved by ensuring that for each $i \in
\{2, \LL, n\}$ and each eliminated strategy $s''_i \in (S_i \cup S'_i)
\setminus \{s_i, s'_i\}$ is not a best response to a joint strategy
$s_{-i}$ in $G$ with $s_1 = a$.
We thus stipulate that 
\[
\mbox{$u_i(a,s_{-1}) := 0$, when $s_i \in (S_i \cup S'_i) \setminus \{s_i, s'_i\}$
for some $i \in \{2, \LL, n\}$.}
\]

To solve the game $(\{a, s_1, s'_1\}, \{s_2, s'_2\}, \ldots, \{s_n, s'_n\}, u_1, \ldots, u_n)$ by IENBR
we first stipulate that the strategies $s_1$ and  $s'_1$ are strictly dominated in it by $a$. We thus stipulate that
\[
\mbox{$u_1(a,s_{-1}) := \max + 1$, when $s_{-1} \in \times_{j \neq 1} \{s_j, s'_j\}$,}
\]
where $\max$ is the maximum payoff used in $G_1$ or $G_2$.

To ensure that the game $(\{a\}, \{s_2, s'_2\}, \ldots, \{s_n, s'_n\}, u_1, \ldots, u_n)$ can be solved by IENBR
we further stipulate that for $i \in \{2, \LL, n\}$

\[
\mbox{$u_i(s_i, s_{-i}) > u_i(s'_i, s_{-i})$, where $s_{-i} \in \{a\} \times_{j \neq i, j \neq 1} \{s_j, s'_j\}$.}
\]

Then for any payoff functions $u_1, \LL, u_n$ that satisfy the above conditions the outcome of the IENBR applied to
above game $G$ yields the joint strategy $(a,s_2, \LL, s_n)$.

We now proceed with the following example.
\begin{myexample}[\textit{IENBR} $\not \Rightarrow \textit{Local}_{\bestr}$]

\rm 
First, consider Figure~\ref{fig:IENBR}. 
\begin{figure}[ht]
\centering
$
\def\objectstyle{\scriptstyle}
\def\labelstyle{\scriptstyle}
\xymatrix@!0@C=30pt@R=30pt{
& (0,1,1) \ar@{->}'[d][dd] \ar@{->}[dl]
&& (1,1,1) \ar@{->}[ll] \ar@{->}[dd]\\
(0,0,1) \ar@{->}[rr] \ar@{->}[dd] & &(1,0,1)
\ar@{->}[ur]\ar@{->}[dd]\\
& (0,1,0) \ar@{--}'[r][rr] \ar@{->}[dl]& & (1,1,0) \ar@{->}[dl]\\
(0,0,0) & &(1,0,0) \ar@{->}[ll]
}
$
\caption{\label{fig:IENBR} A best response graph}
\end{figure}
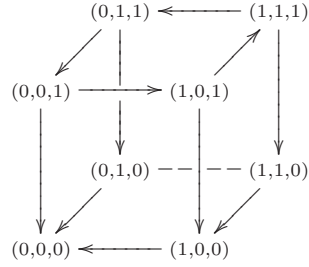

It depicts the best response graph of a three player game $G_1$.
We define the payoff functions analogously to Example~\ref{exa:LnF}.
Note that this game can be solved by IENBR with the outcome being the joint strategy $(0,0,0)$.

Further, this game respects a set-based BR-scheduler $f$ iff
\[
\mbox{$f(\{1,3\}) = 3$ or $f(\{2,3\}) = 3$.}
\]
So it does not respect the local BR-schedulers that correspond to the permutations $1,2,3$ and $2,1,3$.

By renaming the players we obtain two other games, $G_2$ and $G_3$, that both can be solved by IENBR and such
that $G_2$ does not respect the local BR-schedulers that correspond to the permutations $1,3,2$ and $3,1,2$
and $G_3$ does not respect the local BR-schedulers that correspond to the permutations $3,2,1$ and $2,3,1$.
Using the construction explained prior to this example twice we construct then a game $G$ that can be solved by IENBR 
and that does not respect any local BR-scheduler.
\HB
\end{myexample}

So this example shows that the class of finite games that can be solved by IENBR is not contained in $\textit{Local}_{\bestr}$.
We suspect that in fact this class is not contained in $\textit{Set}_{\bestr}$, however were not able to prove it.

\section{Potentials}
To characterize the finite games that have the FIP \cite{MS96}
introduced the concept of a (generalized ordinal) \bfe{potential}. We
now introduce an appropriately modified notion to characterize the
games that respect a scheduler. We shall use it in the next section to
reason about a natural class of games.

Consider a game 
$(S_1, \ldots, S_n, p_1, \ldots, p_n)$ and a 
scheduler $f$. We say that a function $F: S \myra
\mathbb{R}$ is an \bfe{$f$-potential} 
(respectively, an \bfe{$f$-BR-potential})
if for every initial prefix of an improvement path
(respectively, an BR-improvement path)
$(\strprofile^1,\LL, \strprofile^k,\strprofile^{k+1})$ that respects
$f$ we have
$F(\strprofile^{k+1}) > F(\strprofile^k)$.

Note that when $f$ is a state-based scheduler, then a function $F$ is an $f$-potential
iff for all $i, s'_i$ and $s$
\[
\mbox{if $f(s) = i$ and $p_i(s'_i, s_{-i}) > p_i(s_i, s_{-i})$, then $F(s'_i, s_{-i}) > F(s_i, s_{-i})$,}
\]
and similarly for the $f$-BR-potentials.
In the proof below we use the following classic result of \cite{Kon27}.
\begin{lemma}[K\"{o}nig's Lemma] 
Any finitely branching tree is either finite or it has an infinite path.  
\end{lemma}

\begin{theorem}
\label{thm:fpotential}
Consider a finite game $G$.
\begin{enumerate}[(i)]
\item $G$ respects a scheduler $f$ iff an $f$-potential exists.  

\item $G$ respects a BR-scheduler $f$ iff an $f$-BR-potential exists.  
\end{enumerate}
\end{theorem}

\begin{proof}
$(i)$ $(\Ra)$ Consider a branching tree the root of which has all
  joint strategies as successors, the non-root elements of which are
  joint strategies, and whose branches are the improvement paths that
  respect $f$.  Because the game is finite this tree is finitely
  branching.  By K\"{o}nig's Lemma this tree is finite, so we conclude
  that the number of improvement paths that respect $f$ is finite.
  Given a joint strategy $s$ define $F(s)$ to be the number of
  improvement sequences (in the sense of the proof of Theorem
  \ref{thm:sch-BR}) that respect $f$ and that terminate in $s$.
  Clearly $F$ is an $f$-potential.

\smallskip

\NI
$(\La)$ Let $F$ be an $f$-potential.  Suppose by contradiction that an
infinite improvement path that respects $f$ exists.  Then the
corresponding values of $F$ form a strictly increasing infinite sequence.
This is a contradiction, since there are only finitely many joint
strategies.

The proof of $(ii)$ is analogous.
\HB
\end{proof}

The argument given in $(i)$$(\Ra)$ follows the proof of \cite{Mil96}
of the fact that every game that has the FIP has a generalized ordinal
potential. Note that when the range of the $f$-potential is finite the
implications $(\Leftarrow)$ in $(i)$ and $(ii)$ also hold for infinite
games.

\section{Cyclic coordination games}
\label{sec:coordination}
In coordination games the players need to coordinate their strategies
in order to choose among multiple pure Nash equilibria. Here we
consider a natural set up according to which the players are arranged
in a directed simple cycle and the payoff functions can yield three
values: 0 if one chooses a `noncommitting' strategy, 1 if one
coordinates with the neighbour and $-1$ otherwise. We call such games
\bfe{cyclic coordination games}. They are special cases of strategic
games introduced in~\cite{SA12} that are naturally associated with
social networks built over arbitrary weighted directed graphs.

More precisely, let $\strgame_{\coord}=(S_1, \ldots, S_n, p_1, \ldots,
p_n)$ be a (possibly infinite) strategic game in which there is a
special strategy $t_0 \in \bigcap_{i \in N} S_i$ common to all the
players. For $i \in N$, let $i \oplus 1$ and $i \ominus 1$ denote the
increment and decrement operations done in cyclic order within
\C{\LLn}. That is, for $i \in \C{1,\LL,n-1}$, $i \oplus 1=i+1$, $n
\oplus 1=1$, for $i \in \C{2,\LL,n}$, $i \ominus 1=i-1$, and $1
\ominus 1=n$. The payoff functions are defined as,

\[
p_i(s) := \begin{cases}
        \phantom{-}0    & \text{if }\  s_{i} = t_0, \\
        \phantom{-}1    & \text{if }\  s_{i} = s_{i \ominus 1} \text{ and } s_i \neq t_0,\\
                  -1    & \text{otherwise}.
        \end{cases}
\]

Thus for each player $i$, and joint strategy $s$, the best response of
$i$ is to choose $s_{i \ominus 1}$ if $s_{i \ominus 1} \in S_i$ and
$t_0$ otherwise. Therefore, $s$ is a Nash equilibrium in the game
$G_{\textit{coord}}$ iff it is is of the form $(t, \dots,
t)$. 

We first show that cyclic coordination games are weakly acyclic. More precisely,
we prove the following stronger result.

\begin{theorem}~
\begin{enumerate}[(i)]
\item The game $\strgame_{\coord}$ has the FIP iff $n = 2$ or
  $\cap_{i \in N} S_i = \{t_0\}$.
\item In $\strgame_{\coord}$, starting from each joint strategy there
  exists an improvement path of length at most $n$ and a BR-improvement
  path of length at most $2n-2$.
\end{enumerate}
\end{theorem}
We just noted that $s$ is a Nash equilibrium in the game $G_{\textit{coord}}$
iff it is is of the form $(t, \dots, t)$. So we can alternatively
state item $(i)$ as: The game $G_{\textit{coord}}$ has the FIP iff $n
= 2$ or it has exactly one Nash equilibrium.

\begin{proof}
$(i)$ $(\Rightarrow)$
As already mentioned when $n =2$, $\strgame_{\coord}$ has the FIP. If $n > 2$, then
the above example implies that $\cap_{i \in N} S_i = \{t_0\}$.
\II

\NI $(\Leftarrow)$ Suppose that $\strgame_{\coord}$ does not have the FIP.
Consider an infinite improvement path $\xi$.  Some player, say $i$, is
selected in $\xi$ infinitely often.  This means that player $i$
selects in $\xi$ some strategy $t \neq t_0$ infinitely often.  Indeed,
otherwise from some moment on in each joint strategy in $\xi$ his
strategy would be $t_0$, which is not the case.

Each time player $i$ switches in $\xi$ to the strategy $t$, the
strategy of his predecessor $i \ominus 1$ is necessarily $t$, as
well. So also player $i \ominus 1$ switches in $\xi$ to $t$
infinitely often.  Iterating this reasoning we conclude that each player
selects in $\xi$ the strategy $t$ infinitely often.
In particular $t \in \cap_{i \in N} S_i$.

\smallskip

\NI$(ii)$ Take a joint strategy $s$. Note that if all payoffs in $s$ are $\geq
0$, then $s$ is a Nash equilibrium. Suppose that some payoff in $s$ is
$< 0$. Then repeatedly select the first player in the cyclic order
whose payoff is negative and let him switch to $t_0$. After at most
$n$ steps the Nash equilibrium $(t_0, \LL, t_0)$ is reached.

For the BR-improvement path we use the local scheduler $\sch$
associated with the identity permutation, i.e., we repeatedly schedule
the first player in the cyclic order who did not select a best
response.

Consider a joint strategy $s$ taken from a BR-improvement path.
Observe that for all $k$ if $s_k \neq t_0$ and $p_k(s) \geq 0$ (so in
particular if $s_k$ is a best response to $s_{-k}$), then $s_k = s_{k
  \ominus 1}$. So for all $i >1$, the following property holds:
\[
\mbox{$Z(i)$: if $\sch(\strprofile) = i$ and $\strprofile_{i-1} \neq
  t_0$ then for all $j \in \{n,1,\LL, i-1\}$, $\strprofile_j =
  \strprofile_{i-1}$.}
\]
In words: if $i$ is the first player who did not select a best
response and player $i-1$ strategy is not $t_0$, then this is a
strategy of every earlier player and of player $n$.
  
Along each BR-improvement path that respects $\sch$ the value of
$\sch(\strprofile)$ strictly increases until the path terminates or at
certain stage $\sch(\strprofile) = n$. Note that then
$\strprofile_{n-1}=t_0$ since otherwise on the account of property
$Z(n)$ all players' strategies are equal, so $\strprofile$ is a Nash
equilibrium and hence $\sch(\strprofile)$ is undefined. So the unique
best response for player $n$ is $t_0$. This switch begins a
new round with player 1 as the next scheduled player.  Player 1 also
switches to $t_0$ and from now on every consecutive player switches to
$t_0$, as well. The resulting path terminates once player $n-2$
switches to $t_0$.

Consequently the length of the generated BR-improvement path is at most $2n-2$.
\HB
\end{proof}

The proof of $(ii)$ shows that each cyclic coordination game respects
a specific state-based scheduler and the local BR-scheduler associated
with the identity permutation.  However, as the following theorem
shows, a much stronger result holds.

\begin{theorem}
\label{thm:coord-local}
Each coordination game $\strgame_{\coord}$ respects every local scheduler.
\end{theorem}

\begin{proof}
For $n=2$, it is easy to see that $\strgame_{\coord}$ has the FIP and hence the
result follows. Therefore, assume that $n>2$.  We prove the result by
showing that for every local scheduler $\sch$, it is possible to
associate an $\sch$-potential with the game $\strgame_{\coord}$.

Let $\sch$ be a local scheduler. By Proposition
\ref{prop:local-permutation}, the choice function $g$ associated with
$\sch$ is of the form $[\pi]$ for some permutation $\pi$ of $1, \LL,
n$. Let $l=\pi(n)$ be the last element in the permutation $\pi$ (this
will be the only information about $\pi$ that we shall rely upon). Let
$U := \{-1,0,1\}^n$ and let $F: S \to U$ be a function defined by
$F(\strprofile) := (\payoff_l(\strprofile), \payoff_{l\oplus
  1}(\strprofile), \payoff_{l \oplus 2}(\strprofile), \ldots,
\payoff_{l \oplus (n-1)}(\strprofile))$.

For $x \in U$ and $i \in \{1,\ldots,n\}$, $x_i$ denotes the $i$-th
entry in $x$ and as before, $x_{-1}=(x_2,\ldots,x_n)$. We also use the
notation $F(\strprofile)[i]$ to denote the $i$-th entry in the
$n$-tuple $F(\strprofile)$.

Let $\ordlexic$ be the strict counterpart of the lexicographic
ordering over the $(n-1)$-tuples of $-1,0,1$, where $-1 \ordlexic 0
\ordlexic 1$. We extend $\ordlexic$ to a relation $\ordtuple \subseteq
U \times U$. For $x,y \in U$ such that $x \neq y$, $x \ordtuple y$ if
one of the following mutually exclusive conditions holds:

\begin{itemize}
\item[C1] $x_1 \in \{-1,1\}$ and $y_1=0$,
\item[C2] $x_1=y_1=0$ and $x_{-1} \ordlexic y_{-1}$,
\item[C3] $x_1,y_1 \in \{-1,1\}$ and $x_{-1} \ordlexic y_{-1}$,
\item[C4] $x_1,y_1 \in \{-1,1\}$, $x_{-1}=y_{-1}$ and $x_{1} \ordlexic
  y_{1}$.
\end{itemize}

In other words, if the first entry of $y$ is $0$ and that of $x$ is
not $0$, then $x \ordtuple y$. If the first entry of both $x$ and $y$
is $0$, then to order $x$ and $y$ we use the lexicographic ordering
over the $(n-1)$-tuples $x_{-1}$ and $y_{-1}$. If the first entry of
both $x$ and $y$ is not $0$, then again to order $x$ and $y$ we use
the lexicographic ordering over $x_{-1}$ and $y_{-1}$, the exception
being when $x_{-1} = y_{-1}$. In this case, to determine the ordering
we use the lexicographic ordering over $x_1$ and $y_1$.

\begin{myclaim}
\label{clm:ordtuple}
The relation $\ordtuple$ is a strict total ordering over $U$.
\end{myclaim}

Assuming Claim \ref{clm:ordtuple}, consider an initial prefix
$\xi_{k+1}=(\strprofile^1,\ldots,\strprofile^k,\strprofile^{k+1})$ of
an improvement path $\xi$ that respects $\sch$. We claim that
$F(\strprofile^k) \ordtuple F(\strprofile^{k+1})$.  We have the
following cases:

\begin{itemize}
\item $\sch(\strprofile^k)= l \oplus i$ where $i \in \{1, \ldots,
  n-1\}$. Since $\xi$ respects $\sch$, we have $\payoff_{l \oplus
    i}(\strprofile^k) < \payoff_{l \oplus i}(\strprofile^{k+1})$, so
  $F(\strprofile^k)[i+1] \ordlexic F(\strprofile^{k+1})[i+1]$. Since
  $i \neq n$, if $i>1$, then by the definition of the payoff functions,
  for all $j \in \{1,\ldots, i-1\}$, $\payoff_{l \oplus
    j}(\strprofile^k) = \payoff_{l \oplus j}(\strprofile^{k+1})$. If
  $i \neq n-1$, then
  $\payoff_l(\strprofile^k)=\payoff_l(\strprofile^{k+1})$ and it
  immediately follows that $F(\strprofile^k) \ordtuple
  F(\strprofile^{k+1})$. Therefore, the interesting case is when
  $i=n-1$. Here we show that the first entry in $F(\strprofile^{k+1})$
  remains $0$ after the update by player $n-1$ iff the first entry in
  $F(\strprofile^{k})$ is $0$.

\begin{itemize}
\item If $F(\strprofile^k)[1]=0$, then
  $F(\strprofile^{k+1})[1]=0$. 

Indeed, suppose $F(\strprofile^k)[1]=0$. Since $\sch(\strprofile^k)\neq l$, we
have $\strprofile^k_l=\strprofile^{k+1}_l$. By the definition of the
payoff function, for any joint strategy $\strprofile$,
$\payoff_l(\strprofile)=0$ iff $\strprofile_l=t_0$. Thus irrespective
of the choice of $l \oplus (n-1)$ we have
$\payoff_l(\strprofile^{k+1})=0$, so $F(\strprofile^{k+1})[1]=0$.

\smallskip

\item If $F(\strprofile^k)[1] \neq 0$, then $F(\strprofile^{k+1})[1]
  \neq 0$. 

Suppose $F(\strprofile^k)[1] \neq 0$. By the definition of the
payoff functions, $\strprofile^k_l \neq t_0$. Since
$\sch(\strprofile^k)\neq l$, we have
$\strprofile^k_l=\strprofile^{k+1}_l$. Therefore irrespective of the
choice of $l \oplus (n-1)$ we have $\payoff_l(\strprofile^{k-1})\neq
0$, so $F(\strprofile^{k+1})[1] \neq 0$.
\end{itemize}

Thus by conditions C2 and C3 in the definition of $\ordtuple$, and the
fact that $(F(\strprofile^k))_{-1} \ordlexic
(F(\strprofile^{k+1}))_{-1}$, we indeed have $F(\strprofile^k)
\ordtuple F(\strprofile^{k+1})$.

\smallskip

\item $\sch(\strprofile^k)=l$. Since $\xi$ respects $\sch$, for all $i
  \in \{1,\ldots,n-1\}$ we have $l \oplus i \in
  \brset(\strprofile^k)$. We claim that in this case, $\strprofile^k_l
  \neq t_0$ and $\strprofile^k_{l \ominus 1}=t_0$. Suppose not. If
  $\strprofile^k_l = t_0$, then for all $i \in \{1,\ldots,n-1\}$, $l
  \oplus i \in \brset(\strprofile^k)$ implies that
  $\strprofile^k_{l\oplus i} = t_0$. This in turn implies that $l \in
  \brset(\strprofile^k)$, which is a contradiction. If
  $\strprofile^k_{l \ominus 1} \neq t_0$, then we have the following
  two possibilities:
  \begin{itemize}
  \item $\strprofile^k_{l \ominus 1}=\strprofile^k_l$. This implies
    $l \in \brset(\strprofile^k)$, which is a contradiction.
  \item $\strprofile^k_{l \ominus 1} \neq \strprofile^k_l$. Then there
    exists $j \in \{1,\ldots, n-1\}$ such that $\strprofile^k_{l\oplus
      j} = \strprofile^k_{l \ominus 1}$ and $\strprofile^k_{l \oplus
      (j-1)} \neq \strprofile^k_{l \oplus j}$. Since
    $\strprofile^k_{l\oplus j} \neq t_0$, this implies that $l \oplus
    j \not\in \brset(\strprofile^k)$, which is a contradiction.
  \end{itemize}


Now, if $\strprofile^k_l \neq t_0$, $\strprofile^k_{l \ominus 1}=t_0$
and $\payoff_l(\strprofile^k)<\payoff_l(\strprofile^{k+1})$, then
$\strprofile^{k+1}_l=t_0$. By C1 in the definition of $\ordtuple$, it
then follows that $F(\strprofile^k) \ordtuple F(\strprofile^{k+1})$.
\end{itemize}

Finally, since the set $U$ which is the range of the function $F$ is
finite and $\ordtuple$ is a strict total order on $U$, we can use an
appropriate encoding $e:U \to \mathbb{R}$ such that $u_1 \ordtuple
u_2$ iff $e(u_1) < e(u_2)$. Then $e(F(\strprofile^k)) <
e(F(\strprofile^{k+1}))$. So $e \circ F$ is an $f$-potential. By the
remark following Theorem \ref{thm:fpotential} the result follows.
\HB
\end{proof}

\noindent {\it Proof of Claim \ref{clm:ordtuple}}. Let $x,y \in U$
such that $x \neq y$. We have the following cases.

\begin{itemize}
\item $x_1 \in \{-1,1\}$ and $y_1=0$. Then by C1, $x \ordtuple y$.


\item $x_1=0$ and $y_1=0$. Then by C2, if $x_{-1} \ordlexic y_{-1}$ then $x \ordtuple y$
  else $y \ordtuple x$.


\item $x_1=0$ and $y_1 \in \{-1,1\}$. Then by C1, $y \ordtuple x$.


\item $x_1,y_1 \in \{-1,1\}$ and $x_{-1} \neq y_{-1}$. Then by C3, if $x_{-1}
  \ordlexic y_{-1}$ then $x \ordtuple y$ else $y \ordtuple x$.


\item $x_1,y_1 \in \{-1,1\}$ and $x_{-1} = y_{-1}$. Then by C4, if $x_1
  \ordlexic y_1$ then $x \ordtuple y$ else $y \ordtuple x$.
\end{itemize}
Further, it can be verified that the relation $\ordtuple$ is
transitive by a straightforward case analysis. 
\HB

\medskip

Note that Theorem \ref{thm:coord-local} cannot be extended to
set-based schedulers. Indeed, suppose that $n > 2$, and for some $t
\neq t_0$ we have $t \in \cap_{i \in N} S_i$. Consider the joint
strategy $s := (t, t_0, \LL, t_0)$ and a set-based scheduler $f$ such
that for all $k \in \{1, \LL, n\}$, 
$f(\{k, k \oplus 1\}) := k \oplus 1, \ f(\{k, k \oplus 2\}) := k$,
with arbitrary values for other inputs.
Then the following infinite improvement path respects this scheduler.
For the  sake
of readability we underlined the strategies that are not best responses.
\[
(\underline{t}, \underline{t_0}, \LL, t_0),  \ (\underline{t}, t,
\underline{t_0}, \LL, t_0),  \
(t_0, \underline{t}, \underline{t_0}, \LL, t_0), \LL
\]


\section{Bounds on finding a Nash equilibrium}
\label{sec:bounds}

Schedulers, apart of allowing us to classify weakly acyclic games are
also useful to improve bounds on finding a Nash equilibrium.  We start with
the following simple observation.

\begin{proposition}\label{prop:imp}
  Consider a weakly acyclic game $G$ for $n$ players in which each
  player has at most $k$ strategies. Then starting from any joint strategy
  there exists a finite improvement path of length at most $k^n$.
\end{proposition}

\begin{proof}
  Given a joint strategy $s$ consider a finite improvement path $\xi$
  that starts from it.  If $\xi$ has a fragment of the form $s^1, \LL,
  s^k, s^1$, we can safely delete from it the fragment $s^1, \LL,
  s^k$.  Repeating this process we get a finite improvement path that
  starts from $s$ in which each joint strategy appears at most once.
  By assumption there are at most $k^n$ joint strategies, which yields the
  claim.
\HB
\end{proof}

In what follows, given a joint strategy $s$ in a finite game $G$ we
focus on a possibly infinite tree ${\cal{T}}(G,s)$ the root of which
is $s$ and the nodes of which are joint strategies such that $s''$ is
a child of $s'$ if $(s',s'')$ is a step in an improvement path
starting from $s$.  We have the following result in which each step
consists of traversing an edge in the tree ${\cal{T}}(G,s)$.

\begin{theorem} \label{thm:sch}
  Consider a weakly acyclic game $G$ for $n$ players in which each
  player has at most $k$ strategies. 

  \begin{enumerate}[(i)]
  \item Suppose that in each joint strategy each player has at most
    one better response.  Then in each tree ${\cal{T}}(G,s)$ a Nash
    equilibrium can be found in $O(n^{k^n})$ steps.
  
\item If $G$ respects a scheduler, then in each tree ${\cal{T}}(G,s)$ a
    Nash equilibrium can be found in $O(k^n)$ steps.
  \end{enumerate}
\end{theorem}

\begin{proof}
 
\NI
$(i)$
By assumption each node in the tree ${\cal{T}}(G,s)$ has at most
$n$ successors.  Consequently, the subtree that consists of the
nodes that lie at a level at most $m$, where $m > 0$, has at most
$\sum_{i = 0}^{m} n^i$ nodes.  So to find a specific node in this
tree takes $O(n^m)$ steps.  Putting $m = k^n$, we get the claim
on the account of Proposition \ref{prop:imp}.

\II

\NI 
$(ii)$ Suppose $G$ respects a scheduler. By Theorem
\ref{thm:sch-positional} $G$ respects a state-based scheduler $f$.
Take an arbitrary joint strategy $s$ and follow an improvement path
$\xi$ that starts in $s$ and respects $f$.  By definition $\xi$ is a
finite path in ${\cal{T}}(G,s)$ that starts from the root. 

Suppose that some strategy appears more than once in $\xi$, i.e.,
$\xi$ starts with a fragment of the form $s^1, \LL, s^i, \LL, s^j,
s^i$, where $1 \leq i \leq j$.  Then, since $f$ is a state-based
scheduler, the infinite improvement path $s^1, \LL, s^{i-1}, (s^i, \LL,
s^j)^*$ also respects $f$, which contradicts the assumption.

So each joint strategy appears at most once in $\xi$. Consequently,
$\xi$ is of length at most $k^n$, which establishes the claim.  
\HB
\end{proof}

In total we introduced four classes of schedulers, the general ones,
state-based, set-based and local.  To define a scheduler we need to
describe its value on all its inputs.  This yields the progression
given in Table \ref{tab:1}. The last entry is justified by Proposition
\ref{prop:local-permutation} that allows us to describe a local
scheduler by specifying a single permutation of the numbers $1, \LL,
n$.  

\begin{table}[hp]
\begin{center}
\begin{tabular}{|l|l|}
\hline
type of scheduler & number of inputs \\
\hline
general           & infinite \\
state-based       & $k^n$ \\
set-based         & $2^n$ \\
local             & $n$ \\
\hline 
\end{tabular}
\\[2mm]
\caption{Input size of various schedulers\label{tab:1}}
\end{center}
\end{table}

\section*{Acknowledgement}

We thank Ruben Brokkelkamp and Mees de Vries for helpful discussions
and Mona Rahn and the reviewers for useful comments.
\bibliographystyle{abbrv}
\bibliography{/ufs/apt/bib/e.bib,/ufs/apt/bib/apt.bib}

\end{document}

%% file: newmacros-a.tex
 \newtheorem{theorem}{Theorem}
 
 \newtheorem{proposition}[theorem]{Proposition}

\newtheorem{lemma}[theorem]{Lemma}
\newenvironment{proof}{\noindent{\bf Proof.}}{\HB}
\newtheorem{myclaim}{Claim}{\bfseries}{\itshape} 
\newtheorem{myexample}{Example}{\bfseries}{\rm} 

\newcommand{\payoff}{p}
\newcommand{\strprofile}{s}

\newcommand{\nat}{\mathbb{N}}

\newcommand{\ordlexic}{\prec_{L}}

\newcommand{\step}[1]{\mbox{$\stackrel{#1}{\Rightarrow}$}}
\newcommand{\betstep}[1]{\mbox{$\stackrel{#1}{\rightarrow}$}}

\newcommand{\strgame}{G}
\newcommand{\bredge}[1]{\step{#1}}
\newcommand{\betredge}[1]{\betstep{#1}}

\newcommand{\sch}{f}
\newcommand{\st}{\mathit{State}}

\newcommand{\bestr}{\mathit{BR}}
\newcommand{\brset}{\mathit{BR}}
\newcommand{\nbrset}{\mathit{NBR}}
\newcommand{\ordtuple}{\lhd}
\newcommand{\coord}{\mathit{coord}}

%% file: macros-apt.tex
\newcommand{\ES}{\emptyset}

\newcounter{symbol}
\newcommand{\indexsyma}[1]%
{\stepcounter{symbol}\index{zzz1 \thesymbol @\protect#1}}
\newcommand{\indexsymb}[1]%
{\stepcounter{symbol}\index{zzz2 \thesymbol @\protect#1}}
\newcommand{\indexsymc}[1]%
{\stepcounter{symbol}\index{zzz3 \thesymbol @\protect#1}}
\newcommand{\indexsymd}[1]%
{\stepcounter{symbol}\index{zzz4 \thesymbol @\protect#1}}
\newcommand{\indexsyme}[1]%
{\stepcounter{symbol}\index{zzz5 \thesymbol @\protect#1}}

\newcommand{\bfe}[1]{\begin{bfseries}\emph{#1}\end{bfseries}\index{#1}}

\newcommand{\myra}{\mbox{$\:\rightarrow\:$}}
\newcommand{\La}{\mbox{$\:\Leftarrow\:$}}
\newcommand{\Ra}{\mbox{$\:\Rightarrow\:$}}
\newcommand{\tra}{\mbox{$\:\rightarrow^*\:$}}

\newcommand{\sse}{\mbox{$\:\subseteq\:$}}

\newcommand{\fa}{\mbox{$\forall$}}
\newcommand{\te}{\mbox{$\exists$}}

\newcommand{\LLn}{\mbox{$1,\ldots,n$}}
\newcommand{\LL}{\mbox{$\ldots$}}

\newcommand{\C}[1]{\mbox{$\{{#1}\}$}}           

\newcommand{\NI}{\noindent}
\newcommand{\HB}{\hfill{$\Box$}}
\newcommand{\VV}{\vspace{5 mm}}

\newcommand{\II}{\vspace{2 mm}}

\newcommand{\szkew}[1]{\relax \setbox0=\hbox{\kern -24pt $\displaystyle#1$\kern 0pt }%
\box0}
{\catcode`\@=11 \global\let\ifjusthvtest@=\iffalse}

\newcounter{oldmycaption}

